\documentclass[aps,pra,twoside,a4paper,onecolumn,11pt,accepted=2022-08-31]{quantumarticle}
\pdfoutput=1
\usepackage[utf8]{inputenc}
\usepackage[english]{babel}
\usepackage[T1]{fontenc}
\usepackage{amsmath}
\usepackage{lipsum}
\usepackage{eufrak}
\usepackage{float}
\usepackage{tikz}
\usepackage{mathtools}
\usepackage{subfigure}
\usetikzlibrary{positioning}
\usepackage{ulem} 
\usepackage{lineno}
\usepackage{dcolumn}
\usepackage{amsmath,amssymb}
\usepackage{bm}
\usepackage{bbm}
\usepackage{overpic}
\usepackage{latexsym}
\usepackage{color}
\usepackage[english]{babel}
\usepackage{latexsym}
\usepackage{psfrag,graphicx}
\usepackage{epsf}
\usepackage{amsmath}
\usepackage{amssymb}
\usepackage{amsfonts}
\usepackage{natbib}
\usepackage{appendix}
\usepackage{verbatim}
\usepackage{enumitem}
\usepackage{dsfont}

\definecolor{mygrey}{gray}{0.35}
\definecolor{myblue}{rgb}{0.2,0.2,0.8}
\definecolor{myzard}{cmyk}{0,0,0.05,0}
\definecolor{mywhite}{rgb}{1,1,1}
\definecolor{myred}{rgb}{0.9,0.1,0.}
\usepackage[colorlinks=true,citecolor=myblue,linkcolor=myblue,urlcolor=myblue]{hyperref}

\usepackage[makeroom]{cancel}
\newtheorem{theorem}{Theorem}
\newtheorem{lemma}[theorem]{Lemma}

\newtheorem{remark}[theorem]{Remark}

\newtheorem{proposition}[theorem]{Proposition}

\newenvironment{proof}{\medskip\noindent\textbf{Proof.}}{\hfill$\blacksquare$\medskip}
\newenvironment{proof-of}[1]{\medskip\noindent\textbf{Proof of {#1}.}}{\hfill$\blacksquare$\medskip}

\newcommand{\norm}[1]{\left\lVert#1\right\rVert}
\newcommand{\ket}[1]{\vert #1 \rangle} 
\newcommand{\bra}[1]{\langle #1 \vert} 

\newcommand{\dketbra}[1]{\ensuremath{| #1 \rangle\!\langle #1 |}}
\newcommand{\braket}[2]{\ensuremath{\langle #1 | #2 \rangle}}
\newcommand{\ketbra}[2]{\ensuremath{| #1 \rangle\!\langle #2 |}}

\vfuzz \hfuzz
\newcommand{\be}{\begin{equation}} 							
\newcommand{\ee}{\end{equation}}
\newcommand{\bea}{\begin{eqnarray}}
\newcommand{\eea}{\end{eqnarray}}
\newcommand{\bematrix}{\left(\begin{matrix}}
\newcommand{\ematrix}{\end{matrix}\right)}
						
\begin{document}

\title[ATI]{The role of coherence theory in attractor quantum neural networks}

\author{Carlo Marconi}
\affiliation{F\'{\i}sica Te\`{o}rica: Informaci\'{o} i Fen\`{o}mens Qu\`{a}ntics, %
	Departament de F\'{\i}sica, Universitat Aut\`{o}noma de Barcelona, 08193 Bellaterra, Spain}
	
	\author{Pau Colomer Saus}
\affiliation{F\'{\i}sica Te\`{o}rica: Informaci\'{o} i Fen\`{o}mens Qu\`{a}ntics, %
	Departament de F\'{\i}sica, Universitat Aut\`{o}noma de Barcelona, 08193 Bellaterra, Spain}

\author{Mar\'ia Garc\'ia D\'iaz}
\affiliation{F\'{\i}sica Te\`{o}rica: Informaci\'{o} i Fen\`{o}mens Qu\`{a}ntics, %
	Departament de F\'{\i}sica, Universitat Aut\`{o}noma de Barcelona, 08193 Bellaterra, Spain}

\author{Anna Sanpera}
\affiliation{F\'{\i}sica Te\`{o}rica: Informaci\'{o} i Fen\`{o}mens Qu\`{a}ntics, %
	Departament de F\'{\i}sica, Universitat Aut\`{o}noma de Barcelona, 08193 Bellaterra, Spain}
\affiliation{ICREA, Pg. Llu\'is Companys 23, 08010 Barcelona, Spain}
\date{30/06/2022}

\begin{abstract}
\noindent We investigate attractor quantum neural networks (aQNNs) within the framework of coherence theory. We show that: i) aQNNs are associated to non-coherence-generating quantum channels; ii) the depth of the network is given by the decohering power of the corresponding quantum map; and iii) the attractor associated to an arbitrary input state is the one minimizing their relative entropy. Further, we examine faulty aQNNs described by noisy quantum channels, derive their physical implementation and analyze under which conditions their performance can be enhanced by using entanglement or coherence as external resources.
\end{abstract}

\maketitle

\noindent Quantum neural networks (QNNs) stem from adding quantum features like correlations, entanglement and superposition to the parallel processing properties of classical neural networks (NNs), an approach which is expected to result in an enhancement of their performances \cite{riste2017demonstration,carleo2019machine,liu2021rigorous}. Trying to implement neural computing (deep learning) with quantum computers results, generically, in an incompatibility, since the dynamics of the former is nonlinear and dissipative, while the latter's is linear and unitary (dissipation can only be introduced by measurements). Nevertheless, a set of desirable properties for QNNs displaying associative memory has been recently proposed \cite{Schuld2014}: i) QNNs should produce an output state which is the closest to the input state in terms of some distance measure; ii) QNNs should encompass neural computing mechanisms such as training rules or attractor dynamics; and iii) the evolution of QNNs should be based on quantum effects.     

Attractor neural networks (aNNs) are a particular class of NNs. They are implemented by a collection of $n$ interacting nodes (artificial neurons) that dynamically evolve towards one of the states of minimal energy of the system \cite{amit1989modeling}. Such metastable states are dubbed \emph{attractors} or \emph{patterns}. Attractor neural networks are used to model associative memory, that is, the capability to retrieve, out of a set of stored patterns, the state which is the closest to a noisy input according to the Hamming distance. Clearly, the larger the number of the attractors, the greater the associative memory, i.e., the storage capacity of the aNN. A paradigmatic example of an aNN is the Hopfield model \cite{hopfield1982}, consisting of a single layer of $n$ artificial neurons, represented by a set of binary variables $\{x_i\}_{i=1}^n$, $x_{i} \in \{\pm 1\}$, which interact pairwise according to a spin glass Hamiltonian.

Desirably, the quantum analogue of aNNs, which we will denote aQNNs, should meet the requirements stated above. Thus, classical bits are here replaced by qubits which evolve under the action of a completely positive and trace-preserving (CPTP) map.
The storage capacity of an aQNN then corresponds to the maximum number of stationary states of such map. An exponential increase of the storage memory of an aQNN, with respect to its classical counterpart, was already shown in \cite{ventura1998} by means of quantum search algorithms. Also, in \cite{rebentrost2018quantum}, the same result was recovered by using a feed-forward interpretation of the quantum Hopfield neural networks (for a recent development of this model see \cite{meinhardt2020quantum,cao2017quantum}). More recently, in \cite{Lewenstein2020}, the explicit form of the CPTP maps possessing the maximal number of stationary states was derived. 

Interestingly, such CPTP maps correspond to non-coherence-generating operations. Such observation motivates our choice of addressing aQNNs from a coherence-theoretic approach. Within this framework, we characterize the properties of aQNNs, such as their physical implementation and the depth of the network, that is, the number of times the map has to be applied to retrieve faithfully the state which is closest to the initial input. We then focus on the realistic scenario of faulty aQNNs, i.e., the case when some error in the realization of the network is taken into account. 

In what follows we briefly review some basic notions regarding both aQNNs and the resource theory of coherence. In Section \ref{results1}, we present our results regarding error-free aQNNs. We show that the evolution of aQNNs with maximal storage capacity is described by a genuinely incoherent operation (GIO). Besides, we demonstrate that, for such aQNNs, the equivalent of the Hamming distance is the quantum relative entropy. 
After deriving their physical implementation, we define their depth and establish a relation to the concept of decohering power. Further, we show that, in the case of noiseless aQNNs, neither coherence nor entanglement can be exploited as resources to enhance their performance. In Section \ref{results2}, we address the above issues in the context of faulty aQNNs, i.e., when we consider the presence of some source of error in the CPTP maps that implement the aQNNs. In this case, we demonstrate that the corresponding aQNNs are described either by strictly incoherent operations (SIOs), or by maximally incoherent operations (MIOs), thus opening the possibility, in the latter case, to an enhancement of their performance by using coherence as an external resource.

\section{Basic concepts}

\subsection{Attractor quantum neural networks}\label{intro_aqnns}

An attractor QNN of the Hopfield type consists of a network of $n$ $d$-dimensional artificial neurons (qudits) which evolve under a quantum channel, i.e., a non-trivial
CPTP map $\Lambda: \mathcal{B}(\mathcal{H}_{in})\rightarrow \mathcal{B}(\mathcal{H}_{out})$. The stored memories correspond to the stationary states of the map, that is, the states $\rho_S$ such that $\Lambda(\rho_S)=\rho_S$ \cite{Lewenstein2020}. For an arbitrary input state $\rho\neq\rho_S$, the successive applications of the map will bring the state to one of the stationary states of the map $\rho_S$.  In what follows, we restrict to the case where $\dim(\mathcal{H}_{in}) = \dim(\mathcal{H}_{out}) = N=d^n$.  As demonstrated in \cite{Lewenstein2020, Lewenstein2022_corr}, a non-trivial CPTP map can have up to $N$ stationary states  $\Lambda(\ketbra{\mu}{\mu})=\ketbra{\mu}{\mu}$, where $\{\ket{\mu}\}_{\mu=0}^{N-1}$ forms an orthonormal basis of $\mathcal{H}_{in}$. Such a map has the form of a generalized decohering map, i.e.,
\begin{align}
\label{lam}
\Lambda(\rho)=\sum_{\mu=0}^{N-1} \rho_{\mu \mu}\ketbra{\mu}{\mu}+ \underset{(\mu < \nu)}{\sum_{\mu, \nu}^{N-1}} \Big[ \rho_{\mu \nu} (1+\alpha_{\mu \nu}) \ketbra{\mu}{\nu} + \mbox{h.c.} \Big]~,
\end{align}
\noindent where $\alpha_{\mu \nu} \in \mathbb{C}$.

To determine the complete-positivity of the map, it is easier to work in terms of its Choi state,  $J_{\Lambda} \in \mathcal{B}(\mathcal{H}_{in} \otimes \mathcal{H}_{out})$, obtained by means of the Choi-Jamio\l kowski isomorphism, so that $\Lambda$ is CPTP  iff $J_{\Lambda} \geq 0$ and $\mbox{Tr}_{out}(J_{\Lambda})= \mathds{1}_{in}$, where $\mbox{Tr}_{out}$ denotes the partial trace over the subsystem $\mathcal{H}_{out}$. The Choi state of the map of Eq.(\ref{lam}) reads 
\begin{equation}
J_\Lambda=\sum_{\mu=0}^{N-1} \ketbra{\mu \mu}{\mu \mu}+\underset{(\mu < \nu)}{\sum_{\mu, \nu}^{N-1}} \Big[(1+\alpha_{\mu \nu}) \ketbra{\mu \mu}{\nu \nu} + \mbox{h.c.} \Big].
\end{equation}
The positivity requirement, $J_\Lambda\geq 0$, demands that $|1+\alpha_{\mu \nu}|^2 \leq 1~, \forall \alpha_{\mu \nu}$  (and $\alpha_{\mu \mu}=0 ~ \forall \mu$), as well as the positivity of all minors of $|J_{\Lambda}|$, which can be checked by, e.g., the Sylvester's condition \cite{Lewenstein2020}. In the most general case checking positivity is evidently hard, but here we simplify our analysis by restricting to the particular cases where $\alpha_{\mu \nu}=\alpha_{\nu \mu} = \alpha \in \mathbb{R}$ for every $\mu \neq \nu$. Upon this requirement, we find that $\Lambda$ is CPTP whenever $\alpha \in [-N/(N-1),0]$. We remark that our results are, nevertheless, general and apply also when this restriction is lifted as long as the map $\Lambda$ is CPTP.
\noindent Throughout this work, we will only consider aQNNs with maximal storage capacity, that is, those whose evolution is given by Eq.(\ref{lam}). With an abuse of language, we will sometimes refer to the map $\Lambda$ in Eq.(\ref{lam}) metonymically as aQNN.

\subsection{Coherence theory}
\noindent In any resource theory, one should firstly introduce the sets of free states and free operations. Given a Hilbert space $\mathcal{H}$ of dimension $N$, we denote by $\mathcal{B}(\mathcal{H})$ the set of the bounded operators acting on $\mathcal{H}$. 
The set of free states in the resource theory of coherence, denoted as $\mathbb{I}$, comprises the so-called \emph{incoherent states}, that is, all the states $\delta\in\mathcal{B}(\mathcal{H})$ that are diagonal in a fixed basis  $\{\ket{i}\}_{i=0}^{N-1}$ of $\mathcal{H}$, i.e., \mbox{$\mathbb{I}=\left\{ \delta = \sum_i \delta_i\dketbra{i} \mid \sum_i \delta_i=1\right\}$}. Free operations are the CPTP maps, $\mathcal{E}$, that leave incoherent states incoherent, i.e., $\mathcal{E}(\mathbb{I})\subset \mathbb{I}$. Stated differently, $\mathcal{E}$ fulfills $\Delta \circ \mathcal{E} \circ \Delta =\mathcal{E} \circ \Delta$, where $\circ$ denotes the composition between two maps and $\Delta$ is the complete-dephasing map in the chosen basis, i.e., $\Delta(\cdot)=\sum_i \ketbra{i}{i}\cdot \ketbra{i}{i}$ \cite{liu2017}. 
Operations satisfying the above relation are said to be \emph{non-coherence-generating}, since they are unable to create coherence on any incoherent state. In contrast to what happens in the resource theories of asymmetry, athermality or entanglement \cite{chitambar2016}, in coherence theory the set of free operations is not unique. This can be grasped by looking at the Kraus structure of the corresponding CPTP maps ($\mathcal E(\cdot)=\sum_{\alpha} K_{\alpha} \cdot K_{\alpha}^{\dagger}$). 

The largest class of non-coherence-generating operations are the \emph{maximally incoherent operations} (MIOs), whose Kraus operators, $\{K_\alpha\}$, fulfill $\sum_{\alpha} K_{\alpha} \mathbb I K_{\alpha}^{\dagger}\subset\mathbb I$ \cite{aberg2006}. A subset of MIOs are the \emph{incoherent operations} (IOs) \cite{baumgratz2014}, consisting of all MIOs whose Kraus operators satisfy the relation $K_\alpha \mathbb{I} K_\alpha ^\dagger \subset \mathbb{I}$ 
for all $\alpha$. Inside the set of IOs we find the \emph{strictly incoherent operations} (SIOs) \cite{winter2016}, for which the Kraus operators further fulfill that $K_\alpha^\dagger \mathbb{I} K_\alpha  \subset \mathbb{I}$ for all $\alpha$. Finally, \emph{genuinely incoherent operations} (GIOs) \cite{DeVicente2017} are SIOs preserving every incoherent state, i.e., $\mathcal{E}_{\text{GIO}}(\delta)=\delta$ for all $\delta\in \mathbb{I}$. As a consequence, GIOs 
present diagonal Kraus operators.

Besides free states and free operations, one should also introduce a proper measure of the resource considered. To quantify the amount of coherence present in an arbitrary state $\rho\in\mathcal{B(H)}$, a  \emph{coherence measure} \cite{baumgratz2014} must be defined as a functional $C:\mathcal{B(H)}\to\mathbb{R}_{\geq0}$ satisfying two main conditions: (i) faithfulness, meaning that $C(\delta)=0$ for all incoherent $\delta\in\mathbb{I}$, 
and (ii) monotonicity, i.e., $C(\rho)\geq C(\mathcal{E}(\rho))$, for all non-coherence-generating operations $\mathcal{E}$.  Among the most typical coherence measures we find the robustness of coherence \cite{napoli2016}, the relative entropy of 
coherence, i.e.,
\begin{equation}
    C_{r.e.}(\rho)=S(\Delta(\rho))-S(\rho),
\end{equation}
where $S(\rho)=-\mbox{Tr} (\rho\log\rho)$ is the Von Neumann entropy \cite{baumgratz2014}, and the $l_1$-coherence measure,
\begin{equation}
    C_{l_1}(\rho)=\sum_{\mu \neq \nu}|\rho_{\mu \nu}|,
\end{equation}  
which is a valid measure under IOs, but not MIOs \cite{bu2017}. 

Finally, every coherence measure achieves its maximum value on the set of \emph{maximally coherent states} ($S_{\text{MCS}}$), defined, in dimension $N$, as $S_{MCS}:=\{ \frac{1}{\sqrt{N}}\sum_{j=0}^{N-1}e^{i\theta_j} \ket{j} \mid \theta_j\in [0,2\pi)~ \forall j\}$ \cite{peng2016}.

\section{Error-free aQNNs}\label{results1}
\subsection{aQNNs as non-coherence-generating operations}\label{ncg}

A direct inspection of Eq.(\ref{lam}) shows that, in the basis $\{\ket{\mu}\}_{\mu=0}^{N-1}$, the condition $\Delta \circ \Lambda \circ \Delta =\Lambda \circ \Delta$ holds, implying that aQNNs are not able to generate coherence on any input state. In particular, since $\Lambda(\delta)=\delta$ for all $\delta \in \mathbb{I}$, it follows that the set of attractors of the aQNN is equivalent to the set of incoherent states $\mathbb{I}$, and that:
\begin{remark}
aQNNs are described by GIOs.
\end{remark}
As stated in the introduction, this observation justifies addressing aQNNs from a coherence-theoretic perspective. Moreover, analogously to the case of aNNs, aQNNs are {\it{bona fide}} models for associative memory. Indeed, in the asymptotic limit, they are able to retrieve the stored attractor which is closest to the input state, in terms of their relative entropy. We show this fact in the following lemma: 

\begin{lemma}
After $r\rightarrow \infty$ iterations, an aQNN outputs the stored attractor that minimizes the relative entropy with respect to the input state $\rho$, i.e., $S(\rho ||\lim_{r\rightarrow \infty}\Lambda^r(\rho))=\min_{\delta\in {\mathbb I}} S(\rho || \delta)$, where $S(\rho||\sigma)=\mbox{Tr}(\rho \log \rho)-\mbox{Tr}(\rho \log \sigma)$ is the quantum relative entropy fo $\rho$ with respect to $\sigma$. Equivalently, $C_{r.e.}(\rho)$ quantifies the minimum relative entropy between $\rho$ and the set of the attractors of the aQNN. 
\end{lemma}
\begin{proof}
From Eq.(\ref{lam}) one notices that applying $\Lambda$ a sufficient number of times on an input state $\rho$ results in a complete dephasing of $\rho$, i.e., $\lim_{r\rightarrow \infty}\Lambda^r(\rho)=\Delta(\rho)$. Now, let us write the relative entropy between $\rho$ and an incoherent state $\delta$ as $S(\rho || \delta)= S(\Delta(\rho))-S(\rho)+S(\Delta(\rho)||\delta)$. It is immediate to see that
\begin{equation}
    \min_{\delta\in {\mathbb I}} S(\rho||\delta)=S(\Delta(\rho))-S(\rho)+S(\Delta(\rho)||\Delta(\rho))=S(\Delta(\rho))-S(\rho)=C_{r.e.}(\rho)~,
\end{equation}
that is, the minimum relative entropy between an input state $\rho$ and the set of incoherent states (or attractors) is achieved on $\Delta(\rho)$, i.e., the state retrieved by the aQNN after a sufficient number of applications. As proven above, such minimum distance between the input state and the retrieved attractor is quantified by the relative entropy of coherence of the input.  
\end{proof}

\subsection{Physical realization of aQNNs}\label{stinespring}
\noindent Physical operations on a system can always be understood as unitary dynamics and projective measurements on a larger system. Indeed, given a quantum channel $\mathcal{E}:\mathcal{B}(\mathcal{H}) \rightarrow \mathcal{B}(\mathcal{H})$, there always exists an ancillary Hilbert space $\mathcal{A}$ of arbitrary dimension and a unitary operation $U \in \mathcal{B}(\mathcal{H} \otimes \mathcal{A})$ such that
\begin{equation}
\label{stine}
\mathcal{E}(\rho) = \mbox{Tr}_{\mathcal{A}}\left[U (\rho \otimes \ketbra{a_{0}}{a_{0}}) U^{\dagger} \right]~,
\end{equation}
\noindent for any $\rho \in \mathcal{B}(\mathcal{H})$, where $\mbox{Tr}_{\mathcal{A}}$ denotes the partial trace on the subsystem $\mathcal{A}$ and $\ketbra{a_{0}}{a_{0}}$ is the initial state of the ancilla. The corresponding unitary $U$ is known as the Stinespring dilation of the map $\mathcal{E}$ \cite{paulsen2003}. \noindent Thus, aQNNs can be  physically realized by appending an ancillary qudit to the network qudits, letting the composite system evolve under the corresponding Stinespring dilation, and finally discarding the ancilla. Knowing that aQNNs are associated to GIOs allows us to derive the Stinespring dilation  of the former in a straightforward way: 
\begin{proposition}
The Stinespring dilation of an $N$-dimensional aQNN is given by
\begin{equation}
\label{ugio}
	U_{\text{aQNN}} = \sum_{\mu=0}^{N-1}\ketbra{\mu}{\mu}\otimes U_\mu,
\end{equation}
 where $\{\ket{\mu}\}_{\mu =0}^{N-1}$ is an orthonormal basis and $U_\mu$ is a unitary operator such that $U_\mu\ket{a_0}=\ket{c_\mu}$, with  $\{\ket{c_\mu}\}_{\mu =0}^{N-1}$ a set of normalized  states fulfilling $\braket{c_\nu}{c_\mu}=1+\alpha_{\mu \nu}$, $\forall \mu\neq \nu$.
\end{proposition}

\begin{proof}
Let $\{\ket{\mu}\otimes \ket{a_\mu}\}$ be an orthonormal basis of the composite Hilbert space $\mathcal{H} \otimes \mathcal{A}$. In \cite{yao2017} it was proven that the action of the Stinespring dilation of a GIO can be expressed as 
\begin{equation}
	U_{\text{GIO}} (\ket{\mu} \otimes \ket{a_{0}} )= \ket{\mu} \otimes \ket{c_{\mu}}~,
\end{equation}
 where  $\ket{c_{\mu}} = \sum_{i} c^{(i)}_{\mu} \ket{a_{i}}$ and $\{\ket{c_{\mu}}\}$ is a set of normalized but not necessarily orthogonal states.
 Expressing the state $\rho$ in the basis $\{\ket{\mu}\}_{\mu =0}^{N-1}$, i.e., $\rho = \sum_{\mu \nu} \rho_{\mu \nu}\ketbra{\mu}{\nu}$, and making use of Eq.(\ref{ugio}), we find that Eq.(\ref{stine}) takes the form
\begin{equation}
\label{stinegio}
\mathcal{E}_{\text{GIO}}(\rho) = \sum_{\mu \nu} \left( \sum_{k}c^{(k)}_{\mu} \bar{c}^{(k)}_{\nu}\right) \rho_{\mu \nu}  \ket{\mu}\bra{\nu}.
\end{equation}
Let us observe that, due to the normalization of the states $\{\ket{c_{\mu}}\}_{\mu =0}^{N-1}$, it holds
\begin{equation}
\sum_{k}c^{(k)}_{\mu} \bar{c}^{(k)}_{\mu}=1~, \quad \sum_{k}c^{(k)}_{\mu} \bar{c}^{(k)}_{\nu} < 1, ~\forall \mu \neq \nu,
\end{equation}
 so that $\mathcal{E}_{\text{GIO}}(\rho) = \rho$ for any diagonal state $\rho$, and the action of $\mathcal{E}_{\text{GIO}}$ does not increase the value of the off-diagonal elements.\\
 A direct comparison between Eq.(\ref{stinegio}) and the map of Eq.(\ref{lam}) shows that the two maps are equivalent if
\begin{equation}
	1+ \alpha_{\mu \nu} = \sum_{k} c^{(k)}_{\mu} \bar{c}^{(k)}_{\nu}=\braket{c_\nu}{c_\mu}, \quad \forall  \mu \neq \nu,
\end{equation}
which completes the proof.
\end{proof}

\subsection{Depth of aQNNs and decohering power}\label{decoh_power}
\noindent Consider the simple case of a maximally coherent qubit $\ket{\Psi_2}=\frac{1}{\sqrt{2}}(\ket{0}+\ket{1})$ suffering decoherence under the action of an aQNN, i.e.,
$\Lambda(\Psi_2)=\frac{1}{2}\begin{pmatrix}
1& 1+\alpha_{01}\\
1+\bar{\alpha}_{01} & 1
\end{pmatrix}$,
where, from now on, we use the notation  $\Psi:=\ketbra{\Psi}{\Psi}$.
From here it is easy to see that an aQNN with a smaller value of $|1+\alpha_{01}|$ needs to be applied less times on a state in order to completely destroy its coherences. To quantify the ability of operations to cause decoherence, the notion of \textit{decohering power} is invoked. The decohering power of a map $\mathcal{E}:\mathcal{B}(\mathcal{H})\rightarrow \mathcal{B}(\mathcal{H}) $, $\dim(\mathcal{H})=N$, with respect to some coherence measure $C$ was introduced in \cite{mani2015}:
\begin{eqnarray}
D_C(\mathcal{E})&=&\max_{\Psi_N\in S_{\text{MCS}}}(C(\Psi_N)-C(\mathcal{E}(\Psi_N)))\nonumber \\
&=& C(\Psi_N)-\min_{\Psi_N\in S_{\text{MCS}}}C(\mathcal{E}(\Psi_N)).
\end{eqnarray}

When considering the $l_1$-coherence measure we immediately find: 
\begin{proposition}\label{decoh}
The $l_1$-decohering power of an $N$-dimensional aQNN described by the CPTP map $\Lambda$ is given by \begin{equation}
D_{C_{l_1}}(\Lambda)=N-1-\frac{1}{N}\sum_{\mu \neq \nu}|1+\alpha_{\mu \nu}|,
\end{equation}
fulfilling $0\leq D_{C_{l_1}}(\Lambda) \leq N-1$.
\end{proposition}

\noindent We define the depth of an aQNN as the minimum number of times, $r$, that the map $\Lambda$ has to be applied on a state until it becomes stationary up to some tolerable error $\eta$, that is, until the classification process is accomplished with sufficient accuracy (see Fig.\ref{fig:classif}a). At that moment, the coherence of the input is small, i.e., $C_{l_1}(\Lambda^r(\rho))=\eta$, with $0<\eta\ll 1$.
\begin{figure}[H]
  \centering
  \includegraphics[width=0.7\linewidth]{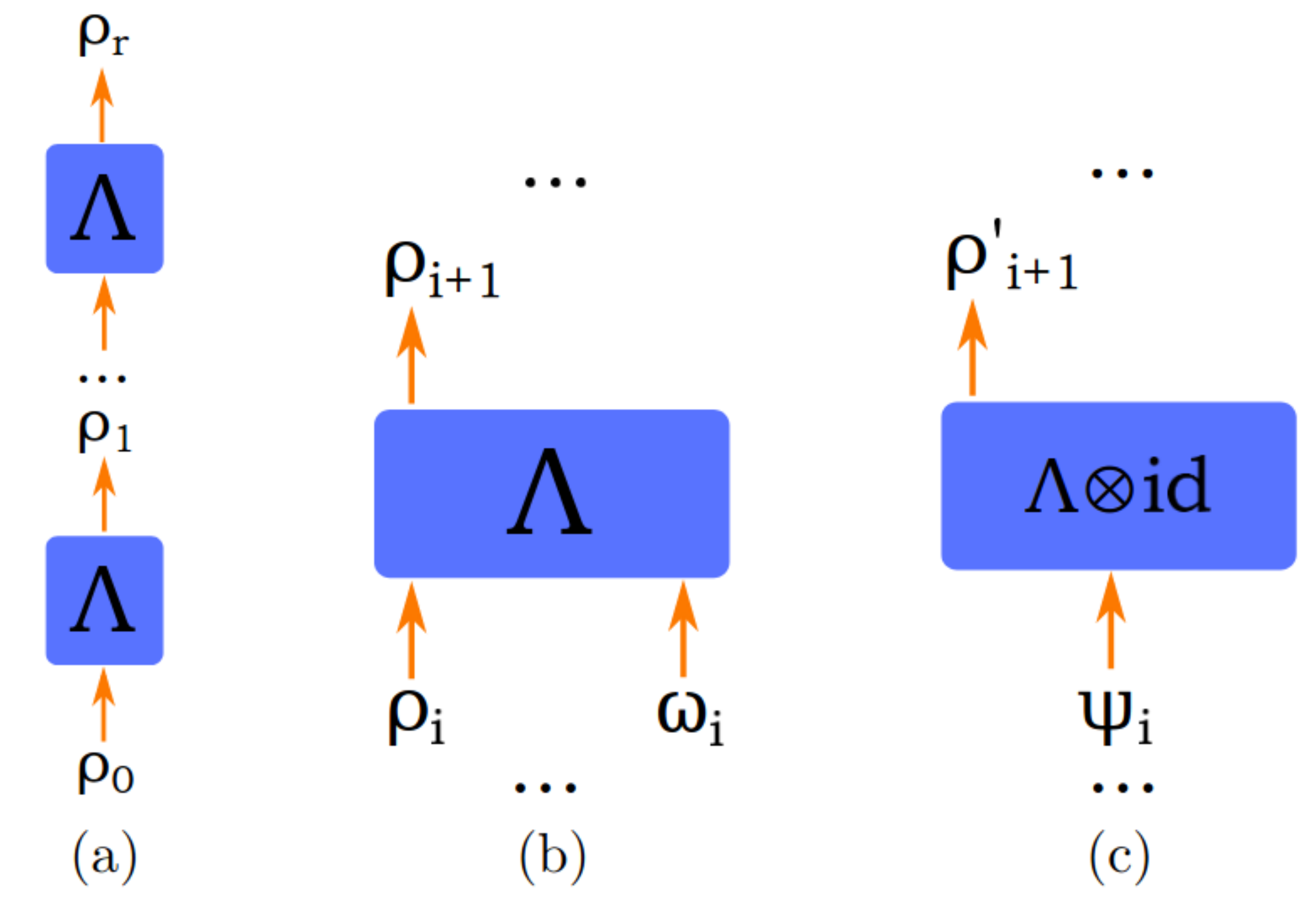}
  \caption{(a) Scheme of a typical classification process. The aQNN described by $\Lambda$ is applied $r$ times until a stationary state $\rho_r$ is reached. (b) Protocol for enhancing the performance of an aQNN associated to $\Lambda$ at layer $i$ using the coherence present in $\omega_i$: $\rho_{i+1}={\cal{N}}_i(\rho_i)=\mbox{Tr}_{\cal A}(\Lambda(\rho_i \otimes \omega_i))$. (c) Protocol for reducing the depth of an aQNN exploiting the entanglement present in $\psi_i$ s.t. $\rho_i=\mbox{Tr}_{\cal A}(\psi_i)$: $\rho'_{i+1}=\mbox{Tr}_{\cal A}\{(\Lambda\otimes \text{id})(\psi_i)\}$.}
       \label{fig:classif}
\end{figure}

Consider the case where the input state is a maximally coherent state, $\Psi_N$, which decoheres uniformly under the action of an aQNN, i.e., $\alpha_{\mu \nu} = \alpha_{\nu \mu} \equiv \alpha ~ \forall \mu,\nu$. In this case we have $C_{l_1}(\Lambda^r(\Psi_N))=(N-1)^{1-r}(N-1-D_{C_{l_1}}(\Lambda))^r$. Allowing for stationarity to be reached within a small error $\eta$, i.e., $C_{l_1}(\Lambda^r(\Psi_N))\leq  \eta$,  immediately yields 

\begin{proposition}
The  depth $r$ of an $N$-dimensional aQNN described by the CPTP map $\Lambda$ with $\alpha_{\mu \nu} = \alpha_{\nu \mu} \equiv \alpha~ \forall \mu \neq \nu$ acting on a maximally coherent state such that stationarity is reached within  $\eta$-precision is given by its $l_1$-decohering power:
\begin{equation}
r\geq \left \lceil{\dfrac{\log(\eta)-\log(N-1)}{\log(N-1-D_{C_{l_1}}(\Lambda))-\log(N-1)}}\right \rceil.
\end{equation}
\end{proposition}
Note that the lower bound for $r$ is tight, since in this case $C_{l_1}(\Lambda^r(\Psi_N))$ and $D_{C_{l_1}}(\Lambda)$ are exactly related.
\begin{figure}[H]
\centering
\includegraphics[width=0.8\linewidth]{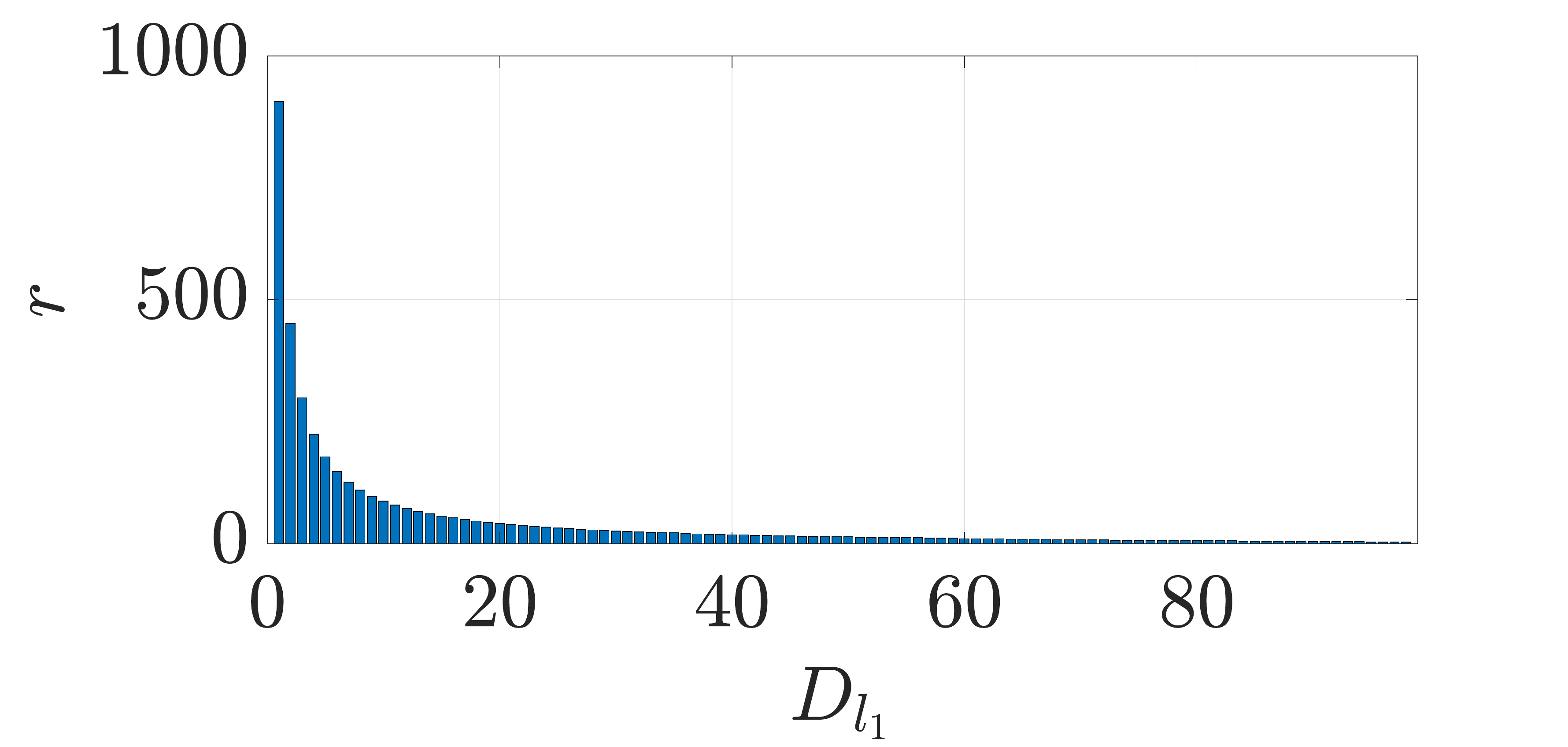}
\caption{Minimum depth of a $100$-dimensional aQNN with $\alpha_{\mu \nu} = \alpha_{\nu \mu} \equiv \alpha$ for all $\mu \neq \nu$ acting on a maximally coherent state such that stationarity is reached within error $\eta=0.01$, as a function of its $l_1$-decohering power.} 
\label{fig:decoh}
\end{figure}

\noindent Fig.\ref{fig:decoh} shows the minimum number of layers that a $100$-dimensional aQNN of this kind needs to have in order for stationarity to be achieved within an error $\eta=0.01$. Moreover, it illustrates how the depth of an aQNN decreases with its decohering power.  
Turning to generic aQNNs, however, it is not possible to find a tight lower bound for the depth, since under non-uniform decoherence the main quantities $C_{l_1}(\Lambda^r(\Psi_N))$ and $D_{C_{l_1}}(\Lambda)$ are not equivalent.

\subsection{Performance of aQNNs cannot be enhanced either by  coherence or entanglement}\label{nogo}

\noindent In both quantum and classical neural computing two main stages are distinguished: the training phase and the inference phase. Throughout this work, we have considered already trained aQNNs and we have investigated their properties during the inference phase. In this section, we are interested in analysing whether the performance of an aQNN can be improved at the inference stage itself. As it is common in the literature about neural computing, enhancing the performance of a neural network implies: i) increasing its accuracy, and ii) accelerating the inference process, that is, reducing the depth of the network (as we defined it in Section \ref{decoh_power}). Here we want to investigate whether the performance of an aQNN can be improved by resorting exclusively to quantum resources. To that aim, one can begin by implementing some channel ${\cal N}_i$ on a given layer $i$ capable of mitigating some of the errors occurred in previous layers, which results in an increased overall accuracy, and/or reducing the number of layers left until the inference process is accomplished, i.e., decreasing the overall $r$. We consider a scenario where an aQNN of arbitrary dimension is coupled to an ancillary system $\cal{A}$, under the action of a CPTP map $\Lambda: \mathcal{B}({\cal{H}} \otimes {\cal{A}}) \rightarrow \mathcal{B}({\cal{H}} \otimes {\cal{A}})$, where the ancilla supplies coherent states $\omega_i\in {\cal  B}({\cal{A}})$.
In this case, one could only exploit the aQNN and the coherent resources in order to realize the target channel ${\cal N}_i:\mathcal{B}({\cal{H}} \otimes {\cal{A}}) \rightarrow \mathcal{B}({\cal{H}})$. The procedure would be as follows: i) append a coherent ancilla $\omega_i\in {\cal B}({\cal{A}})$ to the input state $\rho_i \in {\cal B}({\cal{H}})$, ii) apply   $\Lambda$ on the composite system, and iii) discard the ancillary state (see Fig.\ref{fig:classif}b). Formally, we can express this process as 
\begin{equation}
\rho_{i+1}={\cal{N}}_i(\rho_i)=\mbox{Tr}_{\cal A} (\Lambda(\rho_i \otimes \omega_{i})).
\end{equation}

As discussed in \cite{ben_dana_erratum}, a non-coherence-generating operation ${\cal M}$ is able to realize a coherent channel in this way only if it can activate coherence, i.e., if it fulfills $\Delta \circ {\cal M} \circ \Delta \neq \Delta \circ {\cal M}$. Noting that  GIOs violate this condition \cite{liu2017}, the following no-go result holds:
\begin{proposition}
Coherence cannot be used to enhance the performance of  aQNNs.
\end{proposition}
Since GIOs are unable to exploit the coherence of $\omega_i$ to help implement ${\cal{N}}_i(\rho_i)$ (unlike MIOs  \cite{Diaz2018} or IOs), aQNNs cannot use coherence to boost their own performance. \\

Another strategy to reduce the depth of an aQNN by increasing its decohering power relies on the exploitation of initial correlations \cite{takahashi2021}. Consider an input state $\rho_i \in \mathcal{B}(\mathcal{H})$, purified by the entangled state $\psi_i \in  \mathcal{B}(\mathcal{H} \otimes \mathcal{A})$, i.e, $\mbox{Tr}_{\cal A} (\psi_{i}) =\rho_i$. The question is to find whether using such an entangled input state causes a stronger decoherence in the output, thus reducing the number of times that the map $\Lambda:  \mathcal{B}(\mathcal{H})\rightarrow \mathcal{B}(\mathcal{H})$ has to be applied before the classification task is completed. Stated differently, we want to investigate whether $C(\Lambda(\rho_i))$ is greater than $C(\mbox{Tr}_{\cal A} \{(\Lambda \otimes \text{id}) (\psi_{i})\})$ (see Fig.\ref{fig:classif}c). We hereby show  this is not possible:
\begin{proposition}
Initial entanglement cannot be used to reduce the depth of aQNNs.
\end{proposition}
\begin{proof}
Consider a generally mixed input state $\rho=\sum_{k} p_k \ketbra{\phi_k}{\phi_k}$, where $p_k\in [0,1]$ and the states $\ket{\phi_k}$ are not necessarily orthonormal. A purification of $\rho$ is given by $\ket{\psi}=\sum_{k} \sqrt{p_k} \ket{\phi_k}\ket{k}$, where $\{\ket{k}\}$ is an orthonormal basis of $\mathcal{A}$. Expressing $\rho$ in this basis, i.e., $\rho= \sum_k p_k \sum_{\mu \nu} c_\mu^{(k)}\bar{c}_\nu^{(k)}\ketbra{\mu}{\nu}$, with $c_\xi^{(k)}=\braket{\xi}{\phi_k}$, leads to
\begin{eqnarray}
\mbox{Tr}_{\cal A}\{(\Lambda \otimes \text{id})(\psi)\}&=&\sum_k p_k \bigg[\sum_\mu |c_\mu^{(k)}|^2\ketbra{\mu}{\mu} \nonumber + \sum_{\mu < \nu}\bigg\{ c_\mu^{(k)}\bar{c}_\nu^{(k)}(1+\alpha_{\mu \nu})\ketbra{\mu}{\nu} +\text{h.c.} \bigg\} \bigg]\\
&=& \Lambda(\rho).
\end{eqnarray}
Therefore, $C(\Lambda(\rho))=C(\mbox{Tr}_{\cal A}\{(\Lambda\otimes \text{id})(\psi)\}$ and initial correlations cannot produce a faster decoherence in the input state $\rho$. 
\end{proof}

\section{Faulty aQNNs}\label{results2}
We examine now the realistic scenario of non-error-free aQNNs, that is, the case where some error in the implementation of the network is taken into account. In particular, we denote as \emph{faulty} an aQNN such that the associated map, $\Lambda_{\epsilon}$, preserves the stationary states up to a certain error $\epsilon \in [0,1]$, i.e.,
\begin{equation*}
\Lambda_{\epsilon}(\ketbra{\mu}{\mu}) = (1-\epsilon) \ketbra{\mu}{\mu} + \frac{\epsilon}{N-1} \sum_{\mu < \nu}^{N-1} \ketbra{\nu}{\nu}~.
\end{equation*}

\noindent According to this definition, it is clear that there exist many maps satisfying the above requirement. In what follows we consider one, denoted as $\Lambda_{\epsilon,\gamma}$, whose action over a generic quantum state $\rho$ can be written as
\begin{equation}
\label{lambda}
\Lambda_{\epsilon,\gamma} (\rho) = \sum_{\mu=0}^{N-1} \rho_{\mu \mu} \Big[(1- \epsilon) \ketbra{\mu}{\mu}+\underset{(\mu < \nu)}{\sum_{ \nu=1}^{N-1}} \frac{\epsilon}{N-1} \ketbra{\nu}{\nu} \Big] +\sum_{\mu< \nu}^{N-1} \Big\{ \rho_{\mu \nu} \Big[ (1+\alpha_{\mu\nu}) \ketbra{\mu}{\nu} +\gamma \ketbra{\nu}{\mu}\Big] + \mbox{h.c.}\Big\}~,
\end{equation}
\noindent with $\epsilon \in [0,1]$ and $\gamma \in \mathbb{C}$.
Notice that Eq.(\ref{lambda}) corresponds to a faulty map where $\epsilon$ represents the error on achieving the stationary states and $\gamma$ is a damping factor in the off-diagonal terms. We define a faulty $(\epsilon,\gamma)$-aQNN as that associated to the map $\Lambda_{\epsilon,\gamma}$ of Eq.(\ref{lambda}), whose corresponding Choi state is
\begin{equation}
\label{JCS}
  J_{\Lambda_{\epsilon,\gamma}} = \sum_{\mu=0}^{N-1} \Big[(1- \epsilon) \ketbra{\mu \mu}{\mu \mu}+\underset{(\mu < \nu)}{\sum_{\nu=1}^{N-1}} \frac{\epsilon}{N-1} \ketbra{\nu \nu}{\nu \nu} \Big] +\sum_{\mu< \nu}^{N-1} \Big[ (1+\alpha_{\mu\nu}) \ketbra{\mu \mu}{\nu \nu} +\gamma \ketbra{\mu \nu}{\nu \mu} + \mbox{h.c.} \Big]~,
\end{equation}
\noindent with $|1+\alpha_{\mu \nu}|^2 \leq (1-\epsilon)^2$ for all $\mu \neq \nu$. %

\noindent Notice that Eq.(\ref{JCS}) can be cast as a direct sum, i.e., $J_{\Lambda_{\epsilon,\gamma}}= J \underset{0\leq \mu <\nu \leq N-1}{\bigoplus} J_{\mu \nu}$, with
\begin{align}
\label{j1}
&J =
\begin{pmatrix*}
1-\epsilon  & 1+ \alpha_{01} &  \cdots & 1+ \alpha_{0,N-1} \\
1+ \bar{\alpha}_{01} & 1-\epsilon   &  \cdots & 1+\alpha_{1,N-1}\\
\vdots & \vdots & \ddots & \vdots &  \\
1+\bar{\alpha}_{0,N-1} & 1+ \bar{\alpha}_{1,N-1} & \cdots &1-\epsilon
\end{pmatrix*}~,\\
\label{j2}
&J_{\mu \nu} =
\begin{pmatrix*}
\epsilon/(N-1)  &  \gamma \\
\bar{\gamma}  & \epsilon/(N-1)
\end{pmatrix*}~,
\end{align}
\noindent where the bar symbol denotes the complex conjugation and each $2 \times 2$ matrix $J_{\mu \nu}$ appears with multiplicity $N(N-1)/2$. Following the same arguments of Section \ref{intro_aqnns}, $\Lambda_{\epsilon,\gamma}$ is a CPTP map iff $J_{\Lambda_{\epsilon,\gamma}}\geq 0$. Choosing $|\gamma| \in [0,\epsilon/(N-1)]$ guarantees that $J_{\mu \nu} \geq 0$ for every $\mu \neq \nu$, but finding analytical conditions on the parameters such that $J \geq 0$ is, in general, a cumbersome task. Nevertheless, recalling that $\alpha_{\mu \nu}=\alpha_{\nu \mu} = \alpha \in \mathbb{R}$ for every $\mu \neq \nu$, $\Lambda_{\epsilon,\gamma}$ is a CPTP map whenever $\alpha \in [(\epsilon -N)/(N-1), -\epsilon]$ and $|\gamma| \in [0, \epsilon/(N-1)]$. \\

Notice that, also in the case of faulty maps $\Lambda_{\epsilon,\gamma}$, the non-coherence-generating condition $\Delta \circ \Lambda_{\epsilon,\gamma} \circ \Delta =\Lambda_{\epsilon,\gamma} \circ \Delta$ holds true.
Moreover, it is possible to show that $(\epsilon,\gamma)$-aQNNs correspond to SIOs:

\begin{lemma}
$(\epsilon,\gamma)$-aQNNs are described by SIOs.
\end{lemma}
\begin{proof}
First, we derive the Kraus operators, defined as $K_{i} =
\sqrt{\lambda^{(i)}}~mat(\lambda^{(i)})$, where $\lambda^{(i)}$ is an eigenvalue of the Choi state and $mat(\lambda^{(i)})$ the row-by-row matrix representation of the corresponding eigenvector $\ket{\lambda^{(i)}}$. The diagonalization of $J_{\Lambda_{\epsilon,\gamma}}$ can be made simpler thanks to the direct sum decomposition of Eqs.(\ref{j1})-(\ref{j2}). Notice that the diagonalization of Eq.(\ref{j1}) always yields diagonal Kraus operators. In fact, 
let us consider an eigenvector of the $N \times N$ matrix $J$, say $\ket{\lambda^{(i)}_{J}}=\left( (\lambda^{(i)}_{J})_{00}, (\lambda^{(i)}_{J})_{11}, \dots, (\lambda^{(i)}_{J})_{N-1,N-1} \right)^{T}$. When extending this vector to dimension $N^2$, we need to add $N$ zeroes between each pair of entries, i.e., $\ket{\lambda^{(i)}_{J}}_{ext}=( (\lambda^{(i)}_{J})_{00},\underbrace{0, \dots, 0}_{N},(\lambda^{(i)}_{J})_{11},\dots,(\lambda^{(i)}_{J})_{N-1,N-1} )^{T}$. Hence, when converting the extended eigenvector into a matrix, it is immediate to find that this operation always yields a diagonal Kraus operator, regardless of the particular eigenvector considered. Let us now inspect the eigenvectors of the operator $J_{\mu \nu}$ of Eq.(\ref{j1}). First notice that, for any $0 \leq \mu < \nu \leq N-1$, $J_{\mu \nu}$ can be written in the chosen basis as 
\begin{equation*}
J_{\mu \nu}= \frac{\epsilon}{N-1}\Big(\ketbra{\mu \nu}{\mu \nu}+\ketbra{\nu \mu}{\nu \mu}\Big) +\gamma \ketbra{\mu \nu }{ \nu \mu} + \bar{\gamma} \ketbra{\nu \mu}{\mu \nu}~.
\end{equation*}
 The diagonalization of $J_{\mu \nu}$ yields a couple of eigenvectors of the form $\ket{\lambda^{(\pm)}_{J_{\mu \nu}}}= ( \pm (\lambda_{J_{\mu \nu}})_{0},(\lambda_{J_{\mu \nu}})_{1} )^{T}$. However, differently from the previous case, when extending these vectors to dimension $N^2$, we need to add $N^2-2$ zeroes whose position will vary according to the specific matrix $J_{\mu \nu}$ considered. It is easily found that the zeroes of the extended eigenvector correspond to the elements of the basis of the form $\ket{\mu' \nu'}$ with $\mu',\nu' \neq \mu, \nu$. Thus, the Kraus operators are given by $K_{\mu \nu}=\kappa^{(1)}_{\mu \nu} \ketbra{\mu}{\nu}+ \kappa^{(2)}_{\mu \nu}\ketbra{\nu}{\mu}$~, for some $\kappa^{(i)}_{\mu \nu} \in \mathbb{C}$. For every incoherent state $\delta$ it holds
\begin{equation}
\label{k}
K_{\mu \nu} \delta K^{\dagger}_{\mu \nu} =
|\kappa^{(1)}_{\mu \nu}|^2 \delta_{\nu \nu}\ketbra{\mu}{\mu} +|\kappa^{(2)}_{\mu \nu}|^2 \delta_{\mu \mu} \ketbra{\nu}{\nu}~, 
\end{equation}
and $ K^{\dagger}_{\mu \nu} \delta K_{\mu \nu}$ is obtained by relabelling $\mu \rightarrow \nu$.
Hence, for every $\delta \in \mathbb{I}$, it holds that $K_{\mu \nu} \delta K^{\dagger}_{\mu \nu} \subset \mathbb{I}$, $K^{\dagger}_{\mu \nu} \delta K_{\mu \nu} \subset \mathbb{I}$~,  with $0 \leq \mu < \nu \leq N-1$.
\end{proof}\\

Further, we provide the expression of the distance between the two quantum channels $\Lambda$ and $\Lambda_{\epsilon,\gamma}$.
In order to do so, we introduce the \textit{diamond distance}, denoted by $D_{\diamond}$, which is formally defined, for any pair of CPTP maps, as \cite{watrous13,aharonov1998quantum}
\begin{equation*}
D_{\diamond}(\mathcal{E},\mathcal{F}) = \frac{1}{2} \max_{\rho_{AB}} \norm{(\mbox{id} \otimes \mathcal{E})(\rho_{AB})-(\mbox{id} \otimes \mathcal{F})(\rho_{AB})}_{1}~,
\end{equation*}
\noindent where $\rho_{AB} \in \mathcal{B}(\mathcal{H}_{A} \otimes \mathcal{H}_{B})$ and  $\norm{X}_{1}=\mbox{Tr}{\sqrt{X X^{\dagger}}}$ is the usual trace norm.

Operationally, the diamond distance quantifies how well one can discriminate between two quantum maps. Indeed, it is possible to show that $\mathcal{E}$ and $\mathcal{F}$ become perfectly distinguishable whenever $D_{\diamond}(\mathcal{E},\mathcal{F})=1$ \cite{wilde2011classical}. The computation of the diamond distance between two CPTP maps can be cast as a semidefinite program (SDP) which admits a simple formulation in terms of their Choi states \cite{Diaz2018}, i.e.,
\begin{align}
\label{sdp}
    \min \quad &\lambda\\ \nonumber
    \mbox{s.t.} \quad &Z \geq J_{\mathcal{E}}-J_{\mathcal{F}}\\ \nonumber
    &\lambda \mathds{1}_{A} \geq \mbox{Tr}_{B}(Z)\\ \nonumber
    & Z \geq 0~. \nonumber
\end{align}
Taking into account $\Lambda$ and $\Lambda_{\epsilon,\gamma}$, the solution of the above SDP program does not admit, in general, a simple analytical expression. However, upon suitable conditions, we prove the following result:
\begin{proposition}
\label{prop3}
Let $\alpha_{\mu \nu} = \alpha_{\nu \mu} \equiv \alpha \in \mathbb{R}$ for all $\mu \neq \nu$ and $\gamma = 0$. Then, the diamond distance between $\Lambda$ and $\Lambda_{\epsilon}$ is given by $D_{\diamond}(\Lambda,\Lambda_{\epsilon}) = \epsilon $~.
\end{proposition}
\begin{proof}
Notice that, since the difference between the Choi states yields a diagonal matrix, the SDP program (\ref{sdp}) can be solved by restricting to the diagonal matrices $Z= \mbox{diag}(z_{00}, \dots, z_{N-1,N-1})$ satisfying the constraints $Z \geq J_{\Lambda_{\epsilon,\gamma}}-J_{\Lambda}$ and $\lambda \mathds{1}_{A} \geq \mbox{Tr}_{B}(Z)$. The former condition is easily satisfied by choosing $z_{ii}=\frac{\epsilon}{N-1}$ whenever $\left(J_{\Lambda_{\epsilon,\gamma}}-J_{\Lambda}\right)_{ii}=\frac{\epsilon}{N-1}$ and $z_{ii}=0$ elsewhere. With this choice, we find $\mbox{Tr}_{B}(Z) =\epsilon \mathds{1}$, so that the latter condition reduces to $\lambda \geq \epsilon$. Hence, the minimization over $\lambda$ yields $\epsilon$, which completes the proof.
\end{proof}

Notice that, when restricting to the case of Proposition \ref{prop3}, for $\epsilon = 1$ it is possible to fully discriminate between $\Lambda$ and $\Lambda_{\epsilon}$. Moreover, we have numerically found that, also when $\gamma \neq 0$, Proposition \ref{prop3} holds true, thus implying that the diamond distance is independent of the choice of $\gamma$.

Regarding the physical realization of $(\epsilon,\gamma)$-aQNNs, the following result holds:
\begin{proposition}
\label{prop5}
The Stinespring dilation of an $N$-dimensional $(\epsilon,\gamma)$-aQNN is given by
\begin{equation}
 U_{(\epsilon,\gamma)\text{-aQNN}}= \sum_{\mu}\sum_{k} c_{\mu}^{(k)} \ketbra{\pi_{k}(\mu)}{\mu} \otimes \ketbra{a_n}{a_0}~,
\end{equation}
where $\pi_{k}$ is a permutation function swapping two states $\ket{\mu}$ and $\ket{\nu}~ \forall \mu \neq \nu$, i.e., $\ket{\pi_{k}(\mu)}=\ket{\nu}$, $\ket{\pi_{k}(\nu)}=\ket{\mu}$, and $\{\ket{c_{\mu}^{(k)}}\}_{\mu =0}^{N-1}$ is a set of normalized states fulfilling
\begin{equation}
\begin{split}
&|c_{\mu}^{(0)}|^2 = 1-\epsilon~,\quad  c_{\mu}^{(0)} \bar{c}_{\nu}^{(0)} = 1+\alpha_{\mu \nu}~,\forall \mu \neq \nu~,\\
&c_{\mu}^{(k)}\bar{c}^{(k)}_{\nu}=\gamma,~ \forall k \neq 0~.
\end{split}
\end{equation}
\end{proposition}

\begin{proof}
The Stinespring dilation of a SIO is given by \cite{chitambar2016}
\begin{equation}
\label{stineSIO}
 U_{\text{SIO}}= \sum_{\mu}\sum_{k} c_{\mu}^{(k)} \ketbra{\pi_{k}(\mu)}{\mu} \otimes \ketbra{a_n}{a_0}~,
\end{equation}
\noindent where $\pi_{k}$ is a permutation function labelled by the index $k$ and the coefficients $\{c_{\mu}^{(k)}\}$ are such that the vector $\ket{c_{\mu}}=(c_{\mu}^{(0)}, \cdots, c_{\mu}^{(r)})$ is normalized.
Inserting Eq.(\ref{stineSIO}) in Eq.(\ref{stine}) it is easily found that
\begin{equation}
\mathcal{E}(\rho)= \sum_{\mu \nu}\sum_{k} c_{\mu}^{(k)} \bar{c}_{\nu}^{(k)} \rho_{\mu \nu}\ketbra{\pi_{k}(\mu)}{\pi_{k}(\nu)}~.
\end{equation}
In order to relate the above expression with the one of Eq.(\ref{lambda}), we rewrite it as
\begin{equation}
\label{ph}
\mathcal{E}(\rho)=  \sum_{\mu} \sum_{k} \rho_{\mu \mu } |c_{\mu}^{(k)}|^2 \ketbra{\pi_{k}(\mu)}{\pi_{k}(\mu)} + \sum_{\mu \neq \nu} \sum_{k} \rho_{\mu \nu} c_{\mu}^{(k)} \bar{c}_{\nu}^{(k)} \ketbra{\pi_{k}(\mu)}{\pi_{k}(\nu)}~.
\end{equation}
Let us now denote by $k=0$ the identical permutation that leaves unchanged the elements of the chosen basis, i.e., $\ket{\pi_{0}(\mu)}=\ket{\mu}$ for all $\mu=0, \dots, N-1$.
Hence, the first term of Eq.(\ref{ph}) can be cast as
\begin{equation}
\label{ph2}
\sum_{\mu} \rho_{\mu \mu } |c_{\mu}^{(0)}|^2 \ketbra{\mu}{\mu} +
\sum_{\mu} \sum_{k \neq 0} \rho_{\mu \mu} c_{\mu}^{(k)} \bar{c}_{\nu}^{(k)} \ketbra{\pi_{k}(\mu)}{\pi_{k}(\nu)}~.
\end{equation}
Comparing Eq.(\ref{ph2}) with the diagonal terms in Eq.(\ref{lambda}), we find
\begin{equation*}
|c_{\mu}^{(0)}|^2 = 1-\epsilon~, \quad  c_{\mu}^{(0)} \bar{c}_{\nu}^{(0)} = 1+\alpha_{\mu \nu}~, ~\forall \mu \neq \nu~.
\end{equation*}
To find the rest of the conditions let us rewrite the second term of Eq.(\ref{ph}) as
\begin{equation}
\label{ph3}
\sum_{\mu \neq \nu} \rho_{\mu \nu} c_{\mu}^{(0)} \bar{c}_{\nu}^{(0)} \ketbra{\mu}{\nu} + \sum_{\mu \neq \nu} \sum_{k \neq 0} \rho_{\mu \nu} c_{\mu}^{(k)} \bar{c}_{\nu}^{(k)} \ketbra{\pi_{k}(\mu)}{\pi_{k}(\nu)}~.
\end{equation}
A direct comparison between Eq.(\ref{ph3}) and the off-diagonal terms of Eq.(\ref{lambda}), shows that we need to impose some restrictions on the permutation function. Choosing $\pi_{k}$ to be a swap between any two pair of orthogonal states, i.e., $\ket{\pi_{k}(\mu)}=\ket{\nu}$ and $\ket{\pi_{k}(\nu)}=\ket{\mu}$ with $\mu \neq \nu$, we find:
\begin{equation*}
|c_{\mu}^{(k)}|^2 = \epsilon/(N-1)~, \quad  c_{\mu}^{(k)} \bar{c}_{\nu}^{(k)} = \gamma~, ~\forall k \neq 0~.
\end{equation*}
\end{proof}

\noindent As a consequence of Proposition \ref{prop5}, an error-free aQNN may turn faulty if the unitary operator that physically implements it degrades from $U_{\text{aQNN}}$ to $U_{(\epsilon,\gamma)\text{-aQNN}}$.

We conclude this section by observing that also SIOs are non-coherence-activating operations \cite{liu2017}, which results in the following no-go proposition:
\begin{proposition}
Coherence cannot be used to enhance the performance of $(\epsilon,\gamma)$-aQNNs.
\end{proposition}
In addition, entanglement cannot be exploited either to accelerate the inference process in this case:

\begin{proposition}
Initial entanglement cannot be used to reduce the depth of \mbox{$(\epsilon,\gamma)$-aQNNs}.
\end{proposition}
\begin{proof}
Consider a generally mixed input state $\rho=\sum_{k} p_k \ketbra{\phi_k}{\phi_k}$, where $p_k\in [0,1]$ and the states $\ket{\phi_k}$ are not necessarily orthonormal. A purification of $\rho$ is given by $\ket{\psi}=\sum_{k} \sqrt{p_k} \ket{\phi_k}\ket{k}$, where $\{\ket{k}\}$ is an orthonormal basis of $\mathcal{A}$. Expressing $\rho$ in this basis, i.e., $\rho= \sum_k p_k \sum_{\mu \nu} c_\mu^{(k)}\bar{c}_\nu^{(k)}\ketbra{\mu}{\nu}$, with $c_\xi^{(k)}=\braket{\xi}{\phi_k}$, leads to
\begin{eqnarray}
&\mbox{Tr}_{\cal A}&\{(\Lambda_{\epsilon,\gamma} \otimes \text{id})(\psi)\}=\sum_k p_k \bigg[\sum_\mu |c_\mu^{(k)}|^2[(1-\epsilon)\ketbra{\mu}{\mu} + \sum_{\mu < \nu}\frac{\epsilon}{N-1}\ketbra{\nu}{\nu}\big]\nonumber \\
&+& \sum_{\mu < \nu}\bigg\{ c_\mu^{(k)}\bar{c}_\nu^{(k)}\big[(1+\alpha_{\mu \nu})\ketbra{\mu}{\nu} + \gamma\ketbra{\nu}{\mu}\big]+\text{h.c.} \bigg\} \bigg]= \Lambda_{\epsilon,\gamma}(\rho).
\end{eqnarray}
Therefore, $C(\Lambda_{\epsilon,\gamma}(\rho))=C(\mbox{Tr}_{\cal A}\{(\Lambda_{\epsilon,\gamma}\otimes \text{id})(\psi)\}$ and initial correlations cannot produce a faster decoherence in the input state $\rho$.
\end{proof}

So far, we have considered the case when faulty aQNNs are described by SIOs, showing that neither coherence nor entanglement can be used to enhance the performance of the associated aQNN. Nevertheless, it is possible to show that, when other sources of error are considered, this is not necessarily the case. In particular, we define a map $\Lambda_{\epsilon,\gamma,\lambda}$ defined as:
\begin{equation}
    \Lambda_{\epsilon,\gamma,\lambda}(\rho) = \Lambda_{\epsilon,\gamma}(\rho) + \sum_{\mu<\nu}  \left[\rho_{\mu \nu} \lambda \ketbra{\mu+1}{\nu+1} + \mbox{h.c.}  \right]~,
\end{equation}
\noindent where $\lambda \in \mathbb{C}$ and $\Lambda_{\epsilon,\gamma}(\rho)$ is the map of Eq.(\ref{lambda}).

\noindent It can be checked numerically that $\Lambda_{\epsilon,\gamma,\lambda}$ corresponds to a MIO, but not IO, thus possibly allowing the use of coherence to enhance the performance of the related aQNN, as proven in \cite{Diaz2018}.

\section{Discussion}

\noindent In this work we have shown the usefulness of coherence theory in the characterization of aQNNs of the Hopfield type. Such networks are always described by quantum channels that do not generate coherence. In the case of error-free aQNNs, the associated CPTP maps, $\Lambda$, correspond to genuinely incoherent operations (GIOs), and the network retrieves the stationary state (attractor) which minimizes its relative entropy with the input state. Further, using the concept of decohering power of a channel, we have provided the analytical expression of the depth of the network, that is, the number of layers required to complete a classification task. For error-free aQNNs, coherence theory shows that neither coherence nor entanglement can act as catalysts to improve their performance. Finally, we have studied the effect of errors in these networks. We have found that the associated quantum channels, $\Lambda_{\epsilon, \gamma,\lambda}$, correspond either to strictly incoherent operations (SIOs), or to maximally incoherent operations (MIOs). While GIOs and SIOs do not allow the use of external resources to improve the performance of the related aQNNs, in the case of MIOs an enhancement could be obtained by using an external source of coherence. As such, we believe that our coherence theoretic analysis of aQNNs represents a valuable tool to address the most relevant questions regarding the performance of quantum neural networks. Such novel perspective should also be considered when inspecting other classes of more complex quantum neural networks, where this approach could bring key insights in the characterization of their performance.

\section*{Acknowledgements}
\noindent We thank John Calsamiglia for useful discussions.
We acknowledge financial support from
the Spanish Agencia Estatal de Investigación, PID2019-
107609GB-I00, from Secretaria d’Universitats i Recerca
del Departament d’Empresa i Coneixement de la Generalitat de Catalunya, co-funded by the European Union
Regional Development Fund within the ERDF Operational Program of Catalunya (project QuantumCat, ref.
001-P-001644), and Generalitat de Catalunya CIRIT
2017-SGR- 1127.

\bibliographystyle{unsrtnat}
\bibliography{main21junio}
\end{document}